\documentclass{article}
\usepackage{microtype}
\usepackage{authblk}
\usepackage{amsmath}              
\usepackage{amssymb}              
\usepackage{amsthm}               
\usepackage{enumerate}            
\usepackage{graphicx}             
\usepackage{mathrsfs}             
\usepackage[tight]{subfigure}     
\usepackage{url}                  

\newtheorem{theorem}{Theorem}
\newtheorem{lemma}[theorem]{Lemma}
\newtheorem{corollary}[theorem]{Corollary}

\theoremstyle{definition}
\newtheorem{definition}{Definition}

\newtheorem{remark}{Remark}

\title{Efficient algorithms to decide tightness}

\author[1] {Bhaskar Bagchi}
\author[3] {Benjamin A. Burton}
\author[2] {Basudeb Datta}
\author[2] {Nitin Singh}
\author[3] {Jonathan Spreer}

\affil[1] {Theoretical Statistics and Mathematics Unit, Indian 
Statistical Institute, Bangalore 560\,059, India. bbagchi@isibang.ac.in} 
\affil[2] {Department of Mathematics, Indian Institute of Science, 
Bangalore 560\,012, India. nitin@math.iisc.ernet.in; dattab@math.iisc.ernet.in.} 
\affil[3] {School of Mathematics and Physics, The University of Queensland,
Brisbane QLD 4072, Australia. bab@maths.uq.edu.au; j.spreer@uq.edu.au.}

\begin{document}

\maketitle

\begin{abstract}
Tightness is a generalisation of the notion of convexity: a space 
is tight if and only if it is ``as convex as possible'', given its topological 
constraints. For a simplicial complex, deciding tightness has a straightforward
exponential time algorithm, but efficient methods to decide tightness are
only known in the trivial setting of triangulated surfaces.

In this article, we present a new polynomial time
procedure to decide tightness for triangulations of $3$-manifolds --
a problem which previously was thought to be hard. Furthermore,
we describe an algorithm to decide general tightness in the case of 
$4$-dimensional combinatorial manifolds which is fixed parameter
tractable in the treewidth of the $1$-skeletons of their vertex links, and we present an algorithm
to decide $\mathbb{F}_2$-tightness for weak pseudomanifolds $M$ 
of arbitrary but fixed dimension which is fixed parameter tractable in the 
treewidth of the dual graph of $M$.
\end{abstract}

\bigskip
\noindent
{\bf 1998 ACM Subject Classification}: F.2.2 Geometrical problems and computations; G.2.1 Combinatorial algorithms


\section{Introduction}

The notion of {\em convexity} is a very powerful setting in Mathematics. Many theorems
in many different mathematical fields only hold in the case of a convex base space.
However, in geometry and topology, the concept of convexity has significant limitations: 
most topological spaces simply do not admit a convex representation
(i.e., most topological features of a space, such as a non-trivial
fundamental group, are an obstruction to convexity).
Nonetheless, there exists a distinct intuition that some representations of
such topologically non-trivial spaces look ``more convex'' than others
(e.g., think of a nicely shaped torus compared to a coffee mug with one handle).

The idea of a tight space captures this intuition in a mathematically
precise way which can be applied to a much larger class of topological spaces
than just balls and spheres.
Roughly speaking, a particular embedding of a topological space into some 
Euclidean space $\mathbb{E}^d$ is said to be {\em tight}, if it is ``as convex
as possible'' given its topological constraints.
In particular, a topological ball or sphere is tight if and only if it is convex.

\medskip
Originally, tight embeddings were studied by Alexandrov in 
1938 as objects which minimise total absolute curvature 
\cite{Alexandrov38ClassClosedSurf}. Later, work by Milnor 
\cite{Milnor50RelBettiHypersurfIntGaussCurv}, Chern and 
Lashof \cite{Chern57TotCurvImmMnf}, and Kuiper~\cite{Kuiper59ImmMinTotAbsCurv}
linked the concept to topology
by relating tightness to the sum of the Betti numbers 
of a topological space. The framework was then applied to polyhedral
surfaces by Banchoff \cite{Banchoff70CritPtTheoEmbPoly}, and later fully developed in
the combinatorial setting by K\"uhnel \cite{Kuehnel95TightPolySubm}.

Finally, work by Effenberger \cite{Effenberger09StackPolyTightTrigMnf}, 
and the first and third authors \cite{Bagchi14StellSpheresTightness} made the concept 
accessible to computations.

\medskip
There is a straightforward relaxation of tightness linking it to another
powerful concept in geometry and topology: Morse theory.
In a somehow vague sense, {\em all} Morse functions on a tightly embedded
topological space must be perfect Morse functions; that is, Morse functions which 
satisfy equality in the Morse inequalities (note that this statement
only holds in general, when applied within the right version of Morse-type theory).
In this setting, asking if a given embedding admits {\em at least
one} perfect Morse function is a closely related natural problem.

\medskip
This article deals with tightness in the discrete setting. More precisely,
our topological spaces are represented as simplicial complexes.
In this case, we only have a finite number of (i) essentially distinct embeddings
into Euclidean space and (ii) essentially distinct discrete Morse functions. Hence,
both the problem of deciding tightness of an embedding and finding a
perfect Morse function become (decidable) algorithmic problems. It is well-known that,
for Morse functions in the Forman sense \cite{Forman98MorseTheoCellCompl}
that finding perfect Morse functions 
is NP-hard in general \cite{Joswig04OptMorseMatchings,
Lewiner03TowOptInDiscMorseTheory,Tancer12CollNPComplete}, but fixed parameter
tractable in the treewidth of the dual graph of the triangulation
\cite{Paixao13ParamCompleOfDMT}. The corresponding complexity
for regular simplexwise linear functions \cite[Section 3B]{Kuehnel95TightPolySubm} is still unknown.
Note, however, that both theories are closely related, since both discrete Morse functions and
regular simplexwise linear functions essentially behave like smooth Morse functions.

\medskip
In Section~\ref{sec:poly} of this article, we prove that the dual problem of 
deciding tightness (i.e., deciding whether {\em all} regular simplexwise linear
functions are perfect) has 
a polynomial time solution in the case of $3$-manifolds. Namely, we have the following.

\begin{theorem}
	\label{thm:poly}
	Let $M$ be a combinatorial $3$-manifold with $n$ vertices. There is an
	algorithm to decide whether or not $M$ is tight with running time
	polynomial in $n$.
	Furthermore, the running time of the algorithm is dominated by the
	complexity of computing $\beta_1 (M,\mathbb{F}_2)$.
\end{theorem}

Given the close relation between deciding tightness and finding perfect Morse functions
which, at least in one version, is known to be NP-hard (see above),
this is a surprising result founded on highly non-trivial mathematical
results given in more detail in \cite{Bagchi14Tight3Mflds}.
Furthermore, this polynomial time solution links to a number of
other problems in computational topology (and in particular in the study of
$3$-manifolds) where polynomial time procedures 
are unknown as of today but conjectured to exist. For instance,
finding a perfect Morse function on a $3$-manifold solves specific instances of the $3$-sphere 
and the unknot recognition problems, for both of which no general polynomial time
solution is known, while at least the latter is believed to be polynomial time solvable 
\cite{Hass99UnknotRecNP,Kuperberg11KnottednessNPModGRH}.

\medskip
In Sections~\ref{sec:fpt} and \ref{sec:arbdim}, we give algorithms 
in the more general class of $d$-dimensional combinatorial manifolds for
$d\leq 4$, and $d$-dimensional weak pseudomanifolds for
$d$ arbitrary but fixed, yielding the following results.

\begin{theorem}
	\label{cor:fptmfld}
	Let $M$ be a combinatorial $d$-manifold, $d\leq 4$. Then
	deciding tightness for any field is fixed parameter tractable
	in the treewidth of the $1$-skeletons of the vertex links of $M$.
\end{theorem}

\begin{theorem}
	\label{thm:fptd}
	Let $M$ be a $d$-dimensional weak pseudomanifold, $d$ arbitrary but
	fixed. Then deciding $\mathbb{F}_2$-tightness is 
	fixed parameter tractable in the treewidth of the dual 
	graph of $M$.
\end{theorem} 

\medskip
These results also hold significance for the study of tight
combinatorial manifolds themselves,
which are rare but very special objects (for instance, tight combinatorial manifolds are
conjectured to be strongly minimal, i.e., they are conjectured to contain the minimum 
number of $i$-dimensional faces, for all $i$, amongst all triangulations of this space 
\cite[Conjecture 1.3]{Kuehnel99CensusTight}).

Apart from a few infinite families of
tight spheres, $2$-neighbourly surfaces, sphere products \cite{Kuehnel95TightPolySubm}, 
and tight connected sums of sphere products \cite{Datta12InfFamTightTrig}, 
only sporadic examples of tight combinatorial manifolds are known. However, these sporadic 
examples feature several of the
most fascinating triangulations in the field, including minimal triangulations of the complex
projective plane \cite{Kuehnel83The9VertComplProjPlane}, the K3 surface 
\cite{Casella01TrigK3MinNumVert}, and a triangulation of an $8$-manifold
conjectured to be PL-homeomorphic to the quaternionic plane \cite{Brehm9215VertTrig8Mani}
(see \cite{Kuehnel99CensusTight} for a comprehensive overview of many more examples).

Most of these examples were proven to be tight by applying purely theoretical criteria on 
tightness. This is mainly due to the fact that, despite the algorithmic nature of the
problem in the discrete setting, no efficient algorithm to decide tightness
in general had been known before.
The algorithms in this article are meant to address and partially solve this issue.

\medskip
Much of the theoretical groundwork underlying Theorems~\ref{thm:poly} and \ref{cor:fptmfld}
have very recently been presented in work of the 
first, third and fifth authors in \cite{Bagchi14Tight3Mflds}. For a more 
thorough discussion about tightness we refer the reader to their article.

\subsection*{Acknowledgement}

This work was supported by DIICCSRTE, Australia and DST, India, under the Australia-India 
Strategic Research Fund (project AISRF06660). Furthermore, the third author 
is also supported by the UGC Centre for Advanced Studies.

\section{Preliminaries}

\subsection{Combinatorial manifolds}

Given an (abstract) simplicial complex $C$, i.e., a simplicial complex without 
a particular embedding, the set of faces of $C$ containing a given vertex
$v \in C$ (and all of their subfaces) is called the {\em star of $v$ in $C$}, 
denoted by $\operatorname{st}_C (v)$.
Its boundary, that is, all faces of $\operatorname{st}_C (v)$ which do not contain $v$,
is referred to as the {\em link of $v$ in $C$}, written $\operatorname{lk}_C (v)$. 
An abstract simplicial complex is said to be {\em pure of dimension $d$} if all
of its maximal faces (that is, faces which are not contained in any other face
as a proper subface) 
are of dimension $d$.

A combinatorial $d$-manifold $M$ is an abstract pure simplicial complex of dimension
$d$ such that all vertex links are combinatorial $(d-1)$-dimensional standard $PL$-spheres. 
The definition is understood to be recursive, where a triangulated $1$-sphere 
is simply given by the boundary of an $n$-gon.

The {\it $f$-vector} of $M$ is a $(d+1)$-tuple $f(M)=(f_0, f_1, \ldots f_d)$
where $f_i$ denotes the number of $i$-dimensional faces of $M$. The set of vertices of 
$M$ is denoted by $V(M)$, and the $d$-dimensional faces of $M$ are referred 
to as {\em facets}. 

We call $M$ {\it $k$-neighbourly} if $f_{k-1} = { f_0 \choose k }$, i.e., 
if $M$ contains all possible $(k-1)$-dimensional faces. 
Given a combinatorial manifold $M$ with vertex set $V(M)$ and $W \subset V(M)$,
the simplicial complex
$ M[W] = \{ \sigma \in M \,|\, V(\sigma) \subset W \}, $
i.e., the simplicial complex of all faces of $M$ with vertex set contained in $W$, 
is called the {\em sub-complex of $M$ induced by $W$}.

A pure simplicial complex of dimension $d$ is said to be a {\em 
weak pseudomanifold} if every $(d-1)$-dimensional face is contained in at most 
two facets. Naturally, any combinatorial manifold is a 
weak pseudomanifold. Given a weak pseudomanifold $M$, 
the graph whose vertices represent the facets and whose edges represent 
gluings between the facets along common $(d-1)$-faces of $M$ is called the 
{\em dual graph} of $M$, denoted $\Gamma(M)$. Weak pseudomanifolds 
are the most general class of simplicial complexes for which a dual graph can 
be defined.

\subsection{Tightness}

In its most general form, {\em tightness} is defined for compact connected 
subsets of Euclidean space.

\begin{definition}[Tightness \cite{Kuehnel95TightPolySubm}]
	\label{def:tightness}
	A compact connected subset $M \subset E^d$ is called 
	{\em $k$-tight with respect to a field $\mathbb{F}$} if for every open 
	or closed half space $h \subset E^d$ the induced homomorphism
	\begin{equation*} 
		H_{k} (h \cap M, \mathbb{F}) \to H_{k}(M,\mathbb{F}) 
	\end{equation*}
	is injective. If $M \subset E^d$ is $k$-tight with respect to $\mathbb{F}$ 
	for all $k$, $0 \leq k \leq d$,	it is called {\em $\mathbb{F}$-tight}.
\end{definition}

If a connected subset $M \subset E^d$ is referred to as {\em tight} without specifying a field, 
then it is usually understood that there exists a field $\mathbb{F}$ such that $M$ is
$\mathbb{F}$-tight.

Here and in the following, $H_{\star}$ denotes an appropriate homology theory (i.e., simplicial
homology for our purposes). Now, in the case that our subset is an abstract simplicial complex, 
tightness can be defined as a combinatorial condition.

\begin{definition}[Tightness \cite{Bagchi14StellSpheresTightness,Kuehnel95TightPolySubm}]
	\label{def:tightness2}
	Let $C$ be an abstract simplicial complex with vertex set $V(C)$ 
	and let $\mathbb{F}$ be a field. We say
	that $C$ is tight with respect to $\mathbb{F}$ if {\em (i)} $C$ is connected, 
	and {\em (ii)} for all subsets $W \subset V(C)$, and for all $0 \leq k \leq d$,
	the induced homomorphism 
	$$ H_k (C[W], \mathbb{F}) \to H_k (C, \mathbb{F}) $$
	is injective.
\end{definition}

Note that the above definition does not depend on a specific embedding. 
However, Definition~\ref{def:tightness} can be linked to Definition~\ref{def:tightness2}
by considering the standard embedding of a simplicial complex $C$ into the $(|V(C)|-1)$-simplex.
Also, notice that in Definition~\ref{def:tightness2} tightness is a definition depending on
a finite number of conditions, namely $2^{|V(C)|}$, giving rise to an exponential time 
algorithm to decide tightness. Unless otherwise stated, for the remainder of this article 
we will work with Definition~\ref{def:tightness2}.

\medskip
There are many criteria in the literature on when a simplicial complex is tight.
See \cite{Kuehnel99CensusTight} for a more thorough survey of the field and 
\cite{Bagchi14StellSpheresTightness,Bagchi14Tight3Mflds} for a summary of recent 
developments. In particular, it is a well-known fact that $\mathbb{F}$-orientable
triangulated surfaces are $\mathbb{F}$-tight if and only if they are $2$-neighbourly
\cite[Section 2A]{Kuehnel95TightPolySubm}. 


However, for dimensions greater than two, and for $2$-dimensional complexes different from combinatorial surfaces, 
this observation no longer holds. In particular, for $d > 2$, no easy-to-check characterisation of tightness
for general combinatorial $d$-manifolds is known (see \cite{Bagchi14Tight3Mflds} for a 
full characterisation of tightness of combinatorial $3$-manifolds with respect to fields of odd characteristic). 

Instead, the above condition for combinatorial surfaces generalises to 
combinatorial $3$- and $4$-manifolds in the following way.

\begin{theorem}[Bagchi, Datta]
	\label{thm:bd}
	An $\mathbb{F}$-orientable combinatorial manifold of dimension 
	$\leq 4$ is $\mathbb{F}$-tight if and only if it is
	$0$-tight and $1$-tight with respect to $\mathbb{F}$.
\end{theorem}

This directly follows from Theorem 2.6 (c) in \cite{Bagchi14StellSpheresTightness}. The $\mathbb{F}$-orientability
is necessary to apply Poincar\'e duality to the Betti numbers. 

\medskip
The first two main results of this paper (presented in Sections~\ref{sec:poly} and \ref{sec:fpt})
make use of this fact, yielding {\em (i)} a polynomial time procedure to decide tightness
for combinatorial $3$-manifolds, and {\em (ii)} a fixed parameter tractable algorithm for combinatorial
$4$-manifolds. 

\subsection{Combinatorial invariants}

In this section, we briefly review the definition of the $\sigma$- and $\mu$-vectors
as introduced in \cite{Bagchi14StellSpheresTightness} by the first and third authors. These
{\em combinatorial invariants} build the foundation of a more combinatorial study of 
tightness of simplicial complexes.

\begin{definition}[Bagchi, Datta \cite{Bagchi14StellSpheresTightness}]
	\label{def:sigma}
	Let $C$ be a simplicial complex of dimension $d$. The $\sigma$-vector
	$(\sigma_0, \sigma_1, \dots, \sigma_d)$ of $C$ with respect to a field $\mathbb{F}$
	is defined as 
	\begin{align}
		\sigma_i = \sigma_i(C; \mathbb{F}) & := \sum_{A\subseteq V(C)} 
		\frac{\tilde{\beta}_i(C[A],\mathbb{F})}{\binom{f_0(C)}{\#(A)}},
		\quad 0\leq i \leq d, \nonumber
	\end{align}
	where $\tilde{\beta}_i$ denotes the reduced $i$-th Betti number.
	For $i>\dim(C)$, we formally set $\sigma_i(X;\mathbb{F}) = 0$.
\end{definition}

Definition~\ref{def:sigma} can then be used to define the following.

\begin{definition}[Bagchi \cite{Bagchi14MuVector} and Bagchi, Datta \cite{Bagchi14StellSpheresTightness}]
	Let $C$ be a simplicial complex of dimension $d$. We define
	\begin{align}
		\mu_0 = \mu_0(C; \mathbb{F}) & :=  \sum_{v\in V(C)} 
		\frac{1}{(1+f_0(\operatorname{lk}_v(C)))},\nonumber \\
		\mu_i = \mu_i(C; \mathbb{F}) & := \delta_{i1}\mu_0(C; \mathbb{F}) +  
		\sum_{v\in V(C)} \frac{\sigma_{i-1}(\operatorname{lk}_v(C))}{1+f_0(\operatorname{lk}_v(C))},  \, \, 1 \leq i
		\leq d, \nonumber
	\end{align}
	where $\delta_{ij}$ is the Kronecker delta. 
\end{definition}

Note that $\mu_1(C; \mathbb{F})$ only depends on the $0$-homology of subcomplexes of
$C$ and is thus independent of the field $\mathbb{F}$. Therefore, we sometimes 
write $\mu_1(X)$ instead of $\mu_1(X; \mathbb{F})$.

\begin{lemma}[Bagchi, Datta \cite{Bagchi14StellSpheresTightness}] 
	\label{lemma:2.6}
	Let $M$ be an $\mathbb{F}$-orientable, $2$-neighbourly, combinatorial closed $d$-manifold,
	$d \leq 4$. Then $\beta_1(M; \mathbb{F}) \leq \mu_1(M)$, and equality holds if and only if 
	$M$ is $\mathbb{F}$-tight.
\end{lemma}

This is a special case of a much more general result, but will suffice for the 
purpose of this work. To read more about many recent advances in studying tightness
of simplicial complexes using the framework of combinatorial invariants, see
\cite{Bagchi14Tight3Mflds}.

\subsection{Parameterised complexity and treewidth}

The framework of {\em parameterised complexity}, as introduced by Downey and Fellows~\cite{Downey99ParamComplexity}, 
provides a refined complexity analysis for hard problems. Given a typically NP-hard problem $p$ with input set $A$, a {\em parameter}
$k : A \to \mathbb{N}$ is defined, grading $A$ into input classes of identical parameter values.
This {\em parameterised version} of $p$ is then said to be {\em fixed-parameter tractable (FPT)} with respect to parameter $k$
if $p$ can be solved in $O(f(k)\cdot n^{O(1)})$ time, where $f$ is an arbitrary function 
independent of the problem size $n$ and $k$ is the parameter value of the input. Note that the exponent of 
$n$ must be independent of $k$. In other words, if $p$ is FPT in parameter $k$ then
$k$ (and not the problem size) encapsulates the hardness of $p$.

A priori, a parameter can be many things, such as the maximum vertex degree of a graph, the sum of the Betti numbers of
a triangulation, or the size of the output. However, the significance of such a fixed parameter tractability result 
strongly depends on specific properties of the parameter. In particular, the parameter should be small for interesting 
classes of problem instances and ideally efficient to compute.
In the setting of computational topology, the {\em treewidth} of various graphs associated with triangulations turn out to be
such good choices of parameter \cite{burton14-courcelle,Paixao13ParamCompleOfDMT,Burton12TautAngleStructNPCompl}. 

Informally, 
the treewidth of a graph measures how ``tree-like'' a graph is. The precise definition is as follows.

\begin{definition}[Treewidth]
\label{def:treewidth} 
A tree decomposition of a graph $G$ is a tree $T=(B,E)$ whose vertices $\{B_i \, | \, i \in I\}$ are called \emph{bags}. 
Each bag $B_i$ is a subset of vertices of $G$, and we require that:
\begin{itemize}
	\item every vertex of $G$ is contained in at least one bag ({\em vertex coverage});
	\item for each edge of $G$, at least one bag must contains both its endpoints ({\em edge coverage}); 
	\item the induced subgraph of $T$ spanned by all bags sharing a given vertex of $G$
	must form a subtree of $T$ ({\em subtree property}).
\end{itemize}

The \textit{width} of a tree decomposition is defined as $\max |B_i|-1$, and the 
\textit{treewidth} of $G$ is the minimum width over all tree decompositions. 
\end{definition}

When describing FPT algorithms which operate on tree decompositions, the following
construction has proven to be extremely convenient.

\begin{definition}[Nice tree decomposition]
\label{def:nice_treewidth} 
A tree decomposition $T=(\{B_i \, | \, i \in I\}, E)$ is called a 
{\em nice tree decomposition} if the following conditions are satisfied:
	\begin{enumerate}
	\item there is a fixed bag $B_r$ with $|B_r|=1$ acting as the root
		of $T$ (in this case $B_r$ is called the {\em root bag});
	\item if bag $B_j$ has no children, then $|B_j|=1$ (in this case 
		$B_j$ is called a {\em leaf bag});
	\item every bag of the tree $T$ has at most two children;
	\item if a bag $B_i$ has two children $B_j$ and $B_k$, then 
		$B_i = B_j = B_k$ (in this case $B_i$ is called a {\em join bag});
	\item if a bag $B_i$ has one child $B_j$, then either
		\begin{enumerate}
			\item $|B_i| = |B_j| + 1$ and $B_j \subset B_i$ 
		(in this case $B_i$ is called an {\em introduce bag}), or
			\item $|B_j| = |B_i| + 1$ and $B_i \subset B_j$ 
		(in this case $B_i$ is called a {\em forget bag}).
		\end{enumerate}
	\end{enumerate}
\end{definition}
	
Nice tree decompositions are small and easy to construct by virtue of the following.

\begin{lemma}[\cite{Kloks94Treewidth}]
	Given a tree decomposition of a graph $G$ of width $k$ and $O(n)$ bags, 
	where $n$ is the number of vertices of $G$, we can find a nice 
	tree decomposition of $G$ that also has width $k$ and $O(n)$ bags in 
	time $O(n)$.
\end{lemma}

We make use of nice tree decompositions in Sections~\ref{sec:fpt} and 
\ref{sec:arbdim}.

\section{Vertex links of tight $3$-manifolds}

Before we give a description of our polynomial time algorithm in Section~\ref{sec:poly}, we first need to have 
a closer look at the vertex links of tight combinatorial $3$-manifolds
and a way to speed up the computation of their $\sigma$-vectors.

\begin{theorem}[Bagchi, Datta, Spreer \cite{Bagchi14Tight3Mflds}]
	 \label{theorem:5.7}
	Let $M$ be a tight triangulation of a closed combinatorial $3$-manifold. 
	Then each vertex link of $M$ is a combinatorial $2$-sphere obtained from
	a collection of copies of the boundary of the tetrahedron $S^{2}_4$ and 
	the boundary of the icosahedron $I^2_{12}$ glued together by iteratively
	cutting out triangles and identifying their boundaries.
\end{theorem}

Furthermore, we have the following decomposition theorem for the $\sigma$-vector.

\begin{theorem}[Bagchi, Datta, Spreer \cite{Bagchi14Tight3Mflds}]
	\label{theorem:5.8}
	Let $C_1$ and $C_2$ be induced subcomplexes of a simplicial complex $C$ and 
	$\mathbb{F}$ be a field. Suppose $C = C_1\cup C_2$ and $K = C_1\cap C_2$. 
	If $K$ is $k$-neighbourly, $k\geq 2$, then 
	$$\sigma_i(C;\mathbb{F}) = (f_0 (C)+1)
	\left( \frac{\sigma_i(C_1; \mathbb{F})}{f_0 (C_1)+1} + \frac{\sigma_i(C_2;\mathbb{F})}{f_0 (C_2)+1}
	 - \frac{\sigma_i(K;\mathbb{F})}{f_0 (K)+1} \right)$$
	for $0\leq i\leq k-2$.
\end{theorem}

Moreover, we call a simplicial complex $C$ a \emph{primitive simplicial
complex} if it does {\em not} admit a splitting $C_1 \cup C_2 = C$ such that 
$K = C_1\cap C_2$ is $k$-neighborly, $k \geq 2$.

In particular, this theorem applies to the $\sigma_0$-value of a combinatorial $2$-sphere 
$S$ which can be split into two discs $D_1$ and $D_2$ along a a common triangle $K$ 
(which is $2$-neighbourly), i.e. $S = D_1 \cup_{K} D_2$. In this case we 
write $S = S_1 \# S_2$, where $S_i$ is the $2$-sphere obtained from $D_i$, 
$1 \leq i \leq 2$, by pasting a triangle along the boundary.
Note that whether or not we paste the last triangle into $D_i$ does not
change the $\sigma_0$-value of the construction.
Now, Theorem~\ref{theorem:5.8} together with some initial computations give rise to the following.

\begin{corollary}[Bagchi, Datta, Spreer \cite{Bagchi14Tight3Mflds}]
	\label{coro:5.11}
	Let $k, \ell\geq 0$ and $(k, \ell) \neq (0,0)$. Then 
	$$\sigma_0(\,\, kI^2_{12}\,\#\, \ell S^{2}_4\,\,) = (9k + \ell +3)
	\left (\frac{617}{1716}k + \frac{1}{20}\ell - \frac{1}{4} \right).$$
\end{corollary}

\section{A polynomial time procedure to decide tightness of $3$-manifold triangulations}
\label{sec:poly}

In this section, we give a proof of Theorem~\ref{thm:poly}; that is, we
describe a polynomial time procedure to decide
tightness for combinatorial $3$-manifolds. The running time of the procedure
is dominated by computing the first Betti number of the combinatorial
$3$-manifold which runs in $O(n^6) = O(t^3)$
time, where $n$ is the number of vertices and $t$ is the number of tetrahedra of
the triangulation.

The algorithm will accept any
$3$-dimensional simplicial complex $M$ on $n$ vertices together with
a field $\mathbb{F}$ of characteristic $\chi (\mathbb{F})$. It will then check if $M$ 
is an $\mathbb{F}$-orientable combinatorial manifold and, in the case it is, 
output whether or not $M$ is tight with respect to $\mathbb{F}$.

We use the fact that, by Lemma~\ref{lemma:2.6}, a $2$-neighbourly,
$\mathbb{F}$-orientable combinatorial $3$-manifold is $\mathbb{F}$-tight if and only if
$\beta_1(M; \mathbb{F}) = \mu_1(M)$.

\medskip
First, note that there is an $O(n^2 \log (n))$ procedure to determine whether 
an $n$-vertex $3$-dimensional simplicial complex is an $\mathbb{F}$-orientable
$2$-neighbourly combinatorial manifold: 

\begin{itemize}
	\item Check that $M$ has ${n \choose 2}$ edges
	and $({n \choose 2} -n)$ tetrahedra, which can be done in $O(n^2)$ time.
	\item Check that each of the triangles occurs exactly in two tetrahedra
	and store this gluing information (this can be done in almost 
	quadratic time $O(n^2 \log n)$).
	\item Compute all $n$ vertex links of $M$. Since $M$ is $2$-neighbourly by
	the above, each vertex link must have $n-1$ vertices. Moreover, if $M$ is a 
	combinatorial manifold, each of the vertex links triangulates a $2$-sphere.
	Since an $(n-1)$-vertex $2$-sphere must have $2n -6$ triangles,
	we can either compute these vertex links in $O(n)$ time, or else conclude that $M$ is not a 
	combinatorial $3$-manifold because some vertex link exceeds this size.
	\item Check that each vertex link is a $2$-sphere. Since we know from the
	above that every triangle in $M$ is contained in exactly two tetrahedra,
	and hence each edge in a link is contained in two triangles, 
	this can be done by computing the Euler characteristic of each link,
	again a linear time procedure for each vertex link.
	\item If $\chi (\mathbb{F})$ is odd, use the gluing information
	from above to compute an orientation on $M$ in quadratic time.
	If $\chi (\mathbb{F})$ is even, $M$ will always be $\mathbb{F}$-orientable.
\end{itemize}

Now, if $M$ fails to be a combinatorial $3$-manifold we stop.
If $M$ is a combinatorial $3$-manifold, but either not $2$-neighbourly or
not $\mathbb{F}$-orientable we conclude that $M$ is not tight (note that
a non-orientable manifold can never be tight, see 
\cite[Proposition 2.9 (b)]{Bagchi14StellSpheresTightness}).
	
In case $M$ is $2$-neighbourly and $\mathbb{F}$-orientable we now have to
test whether
$\beta_1(M; \mathbb{F}) = \mu_1(M)$. More precisely, we must compute 
$\sigma_0 (\operatorname{lk}_M (v))$ for all vertices $v \in V(M)$, a procedure which
na\"ively requires the analysis of $O(2^n)$ induced subcomplexes of $\operatorname{lk}_M (v)$.

However, using Corollary~\ref{coro:5.11}, computing the $\sigma_0$-value of
the vertex link of a tight combinatorial $3$-manifold simplifies to the
following algorithm to be carried out for each vertex link $S$ (note that
because of the $2$-neighbourliness, every vertex link needs to contain
$n-1$ vertices, $3n-9$ edges, and $2n-6$ triangles):

\begin{itemize}
	\item Enumerate all induced $3$-cycles of $S$ and split $S$
	along these cycles into separate connected components.
	This can be done in $O(n^3)$ time by first storing a list of adjacent edges for 
	each vertex and a list of adjacent triangles for each edge,
	and then running over all vertex subsets of size 	
	three checking if the induced subcomplex is an empty $3$-cycle.
	\item For each connected component check if the $1$-skeleton of the component 
	is isomorphic to the graph of either $S_4^2$ or $I_{12}^2$.
	If not, return that $M$ is not tight. Otherwise,
	sum up the number of connected components of each type.
	Note that each component can be processed in constant time, but
	that there is a linear number of components.
	\item Use Corollary~\ref{coro:5.11} to compute the $\sigma_0$-value
	of the link.
\end{itemize}

Add up all $\sigma_0$-values and divide by $n$ to obtain $\mu_1 (M)$.
Now, we know from Lemma~\ref{lemma:2.6} that $M$ is tight if and only
if $\mu_1(M) = \beta_1 (M, \mathbb{F})$. Thus, if $\mu_1 (M)$ is not
an integer, $M$ cannot be tight. This step overall requires a running time
of $O(n^3)$.

Finally, suppose $\mu_1 (M)$ is an integer. 
In this case, we must compute $\beta_1 (M,\mathbb{F})$.
For all $\mathbb{F}$, this information is encoded in $H_1 (M,\mathbb{Z})$
which can be computed in $O(n^6)$ by determining the Smith normal form
of the boundary matrix. Hence,
this last step dominates the running time of the algorithm.

\medskip
Altogether, checking for tightness can be done in the same time complexity as
computing homology, a task which is considered to be easy in computational
topology.

Furthermore, recent theoretical results in \cite{Bagchi14Tight3Mflds} show that,
if the characteristic of $\mathbb{F}$ is odd, then $M$ must be what is called a 
$2$-neighbourly, {\em stacked} combinatorial manifold. Such an object
can be identified in $O(n^2 log(n))$ time (i.e., almost linear in the number of 
tetrahedra).

\section{A fixed parameter tractable algorithm for dimension four}
\label{sec:fpt}

In the previous section, we saw that deciding tightness of $3$-manifolds
can be done efficiently. However, the algorithm for dimension three
relies heavily on
special properties of the vertex links. No such characterisation
of the vertex links is known in higher dimensions.

However, both Lemma~\ref{lemma:2.6} and Theorem~\ref{theorem:5.8} can still be applied 
in the $4$-dimensional setting. Hence, any computation of the $\sigma_0$-value 
of the vertex link of some combinatorial $4$-manifold immediately reduces
to computing the $\sigma_0$-value of the primitive components of that
vertex link. For the remaining primitive pieces we have the following.

\begin{theorem}
	\label{thm:fpt}
	Let $C$ be a simplicial complex $C$ whose $1$-skeleton
	has treewidth $\leq k$. Then there exists an
	algorithm to compute $\sigma_0 (C)$ in $O(f(k) \,\, n^5)$ time, where
	$n$ is the number of vertices of $C$.
\end{theorem}

\begin{proof}[Proof of Theorem~\ref{thm:fpt}]
	We give an overview of the structure of this algorithm.

	First of all, note that for any simplicial complex $C$, 
	the value $\sigma_0 (C)$ only depends on the $1$-skeleton $C_1$ of $C$.

	Thus, let $T = (B,E)$ be a nice tree decomposition of
	$C_1$. We write $B = \{ B_1, B_2, \ldots , B_r \}$, where 
	the bags are ordered in a bottom-up fashion which 
	is compatible with a dynamic programming approach. In other words:
	the tree $T$ is 
	rooted; whenever $B_i$ is a parent of $B_j$ we have $i > j$;
	and our algorithm will process the bags in the order $B_1,\ldots,B_r$.

	For each bag $B_i$, we consider the induced subgraph $C_1[B_i]$ 
	spanned by all vertices in the bag $B_i$. Furthermore, we
	denote the induced subgraph of the $1$-skeleton spanned by all 
	visited vertices at step $i$ by 
	$$ C_{1,i}^- := C_1[B_1 \cup B_2 \cup \ldots \cup B_i].$$

	\medskip
	Given a bag $B_i$: for each subset of vertices $S \subset B_i$, for each
	partition $\pi$ of elements of $S$, and for all integers $c,m$ with
	$c \leq m \leq n$,
	we count the number of induced subcomplexes of $C_{1,i}^-$ with
	$m$ vertices and $c$ connected components whose vertex set intersects
	bag $B_i$ in precisely the set $S$, and whose connected components
	partition this set $S$ according to $\pi$.
	Note that the count $c$ includes connected
	components which are already ``forgotten'' (i.e., which do not meet
	bag $B_i$ at all). Here, $n$ denotes the number of vertices of $C$.

	These lists can be trivially set up in constant time for each one-vertex
	leaf bag.

	For each introduce bag, the list elements must be updated by either including
	the added vertex to the induced subcomplex or not. Note that in each 
	step, the edges inside $C_1[B_i]$ place restrictions on which 
	partitions of subsets $S \subset B_i$ can correspond to induced subcomplexes 
	in $C_{1,i}^-$. The overall running time of such an introduce operation
	is dominated by the length of the 
	list, which is at most quadratic in $n$ multiplied by a function in $k$
	(for each subset and partition, up to $O(n^2)$
	distinct list items can exist corresponding to different values of
	$c,m$).

	For each forget bag, we remove the forgotten vertex from each list item, thereby possibly
	aggregating list items with equal values of $S$, $\pi$, $c$ and $m$.
	This operation again has a running time dominated by
	the length of the lists.

	Finally, whenever we join two bags, we pairwise combine list elements
	whenever the underlying induced subcomplexes are well defined (i.e.,
	whenever the subsets in the bag coincide and the partitions are
	compatible).
	This requires $O(n^4)$ time in total (the product of the two child list
	lengths).

	\medskip
	After processing the root bag we are left with $O(n^2)$ list entries, labelled
	by the empty set, the empty partition, and the various possible
	values of $c,m$.
	Given the values of these list items it is now straightforward to 
	compute $\sigma_0 (C)$, as in Definition~\ref{def:sigma}.
	Given that there are overall $O(n)$ bags to process,
	we have an overall running time of $O(f(k) \,\, n^5)$.
\end{proof}

The FPT algorithm to decide tightness for $d$-dimensional combinatorial manifolds
$M$, $d \leq 4$, now consists of computing $\mu_1 (M)$ by applying Theorem~\ref{thm:fpt}
to each vertex link,
and comparing $\mu_1(M)$ to 
$\beta_1 (M,\mathbb{F})$. Since the computation of $\mu_1 (M)$ is independent of
$\mathbb{F}$ and $\beta_1 (M,\mathbb{F}) \leq \mu_1 (M)$ for all $\mathbb{F}$, we
can choose $\mathbb{F}$ to maximise $\beta_1$ and, by Theorem~\ref{thm:bd}, $M$ is
tight if and only if 
$$ \underset{\mathbb{F}}{\max}\,\,\, \beta_1 (M, \mathbb{F}) = \mu_1 (M). $$
Following the proof of Theorem~\ref{thm:fpt}, the overall running time of 
this procedure is $O(\operatorname{poly} (n) + n^6 f(k))$ for some 
function $f$.

\begin{remark}
	Note that the $1$-skeleton of the vertex link of a tight combinatorial
	$4$-manifold $M$ might be the complete graph, which has maximal treewidth.
	In this case, the algorithm presented in this section clearly fails to
	give feasible running times. However, note that in this case $M$
	must be a $3$-neighbourly simply connected combinatorial $4$-manifold
	which is tight by \cite[Theorem 4.9]{Kuehnel95TightPolySubm}.
	Hence, worst case running times for the algorithm above are expected
	for combinatorial $4$-manifolds of high topological complexity and with low,
	but strictly positive, first Betti numbers. 
\end{remark}

\section{An FPT algorithm to decide $\mathbb{F}_2$-tightness in 
arbitrary dimension}
\label{sec:arbdim}

In the following, we present an algorithm to decide 
$\mathbb{F}_2$-tightness of an arbitrary weak pseudomanifold $M$
of fixed dimension $d$. This algorithm is fixed parameter tractable in the 
treewidth of the dual graph $\Gamma(M)$ of $M$.
We restrict our algorithm here to connected complexes, and note that
it can be extended to non-connected complexes by simply processing
each connected component separately. 

\medskip
Let $M$ be a connected $d$-dimensional weak pseudomanifold $M$ 
with dual graph $\Gamma(M)$ of treewidth $\leq k$, and let $T = (B,E)$ 
be a nice tree decomposition of $\Gamma (M)$ of width $\leq k$.
As in the previous section, we order the bags 
$B = \{ B_1, B_2, \ldots , B_r \}$ of our nice tree decomposition
so that their order coincides with how they are processed
by the algorithm and, in particular, such that 
whenever $B_i$ is a parent of $B_j$, then $i > j$.

For each bag, we consider the induced subcomplex $M[V(B_i)]$ spanned by
all vertices of the facets (i.e., $d$-simplices) of $M$ associated to the bag $B_i$. 
It follows that $M[V(B_i)]$  has at most
${ (d+1)(k+1) \choose j+1 }$ faces of dimension $j$. Furthermore, we 
denote the subcomplex spanned by all visited facets at step $i$ by 
$$ M_i^- := M[V(B_1 \cup B_2 \cup \ldots \cup B_i)].$$

First of all, note that $M$ must always be $d$-tight with respect to $\mathbb{F}_2$.
Furthermore, $M$ is $0$-tight if and only if $M$ is $2$-neighbourly. Thus,
$0$- and $d$-tightness can be checked in polynomial time.

To decide $j$-tightness, $1 \leq j \leq d-1$, the idea of the algorithm
is to look for an obstruction to $j$-tightness (with respect to $\mathbb{F}_2$). Following
Definition~\ref{def:tightness2}, such an obstruction is given by an induced subcomplex $N \subset M$,
containing a $j$-cycle which is a boundary in $M$, but not in $N$. 

Now, set $s_{i,\ell} = \operatorname{skel}_{\ell} (M[V(B_i])$,
$0 \leq \ell \leq d$. For each bag $B_i$, 
we store a list of triples $(A,b,\mathcal{C})$, where 
$A \subset s_{i,j}$, $b \subset s_{i,j-1}$ and $\mathcal{C}$
is a list of subsets of $s_{i,j}$ such that {\em (i)} $A = c \cap M[V(B_i)]$ 
where $c$ is a $j$-chain in $M_i^-$ for which $\partial c = b 
\subset M[V(B_i)]$, and {\em (ii)} $\mathcal{C}$ contains all $j$-chains $C$ of
$M[V(B_i)]$ such that there is a $(j+1)$-chain $D$ in $M_i^-$ with
$c \cup  C = \partial D$.
Note that $\mathcal{C}$ might contain the empty set as an element (this is 
the case if $c$ itself is
the boundary of a $(j+1)$-chain) and, in this case, is not considered to be 
empty. See Figure~\ref{fig:triple} for an example of a triple 
$(A,b,\mathcal{C})$.

\begin{figure}
	\begin{center}
	\includegraphics[width=\textwidth]{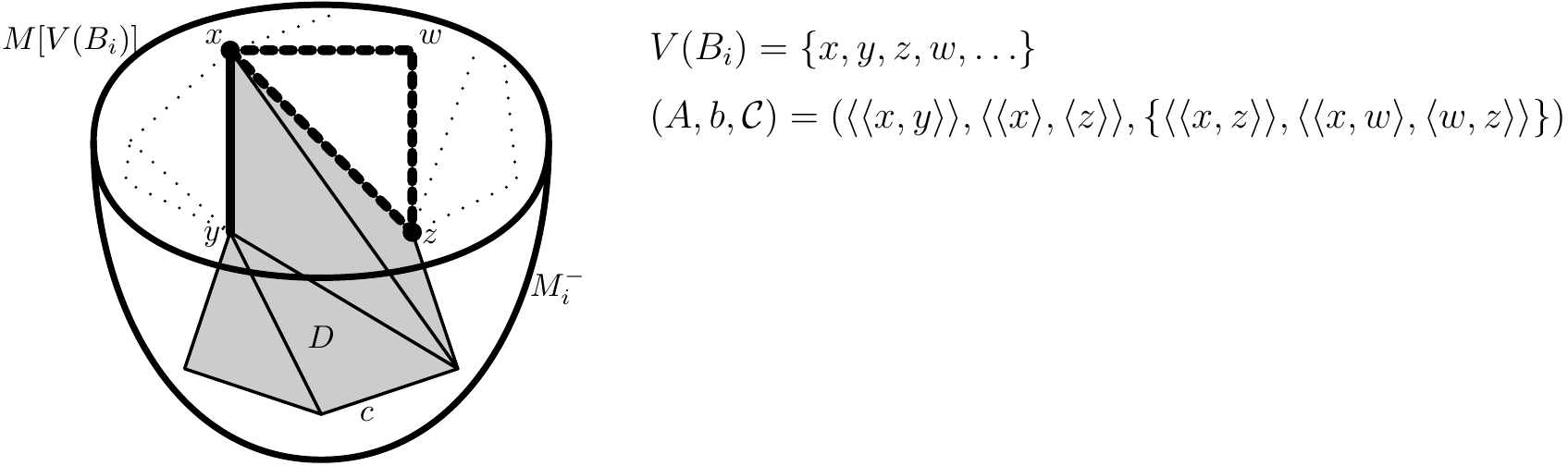}
	\end{center}	
	\caption{Complex $M_i^-$ with ``boundary'' $M[V(B_i)]$,
		an example of a triple $(A,b,\mathcal{C})$, and
		an example of a disc $D$ bounding $c$ and
		an element of $\mathcal{C}$.
		Note that $c$ starts at $x$, goes down to $y$, and 
		passes further down along the grey area and back up to $z$.
		\label{fig:triple}}
\end{figure}

We now give the algorithm by describing all operations performed at
a leaf bag, an introduce bag, a forget bag, and a join bag of our nice tree 
decomposition $T=(B,E)$. Note that, since we are looking at chains over
$\mathbb{F}_2$, chains are defined by subsets of faces. That is, every 
face in a subset is assigned coefficient $1$ in the chain and any face
outside the subset is assigned coefficient $0$.

\subsection{The leaf bag}

Since $T$ is a nice tree decomposition, every leaf bag $B_L$ 
consists of a single vertex of the dual graph. Hence, $M[V(B_L)]$ consists
of a single facet $\Delta$ of $M$. 

We initialise our list for each such bag with triples $(A,b,\mathcal{C})$, 
where $A$ runs through all subsets of $j$-faces of $\Delta$,
$0 \leq j \leq d-1$, and $b = \partial A$ in $\mathbb{F}_2$, i.e., $b$ consists
of all $(j-1)$-faces of $A$ of odd degree.
Now, every $j$-cycle in the $d$-simplex $\Delta$ is also the boundary of a $(j+1)$-chain
and hence $\mathcal{C}$ contains all subsets $C \subset
\operatorname{skel}_j (\Delta) \setminus A$,
such that $A \cup C$ is a $j$-cycle closing off $A$, i.e., all $(j-1)$-faces in
$A \cup C$ have even degree. Furthermore, note that $\mathcal{C}$ contains 
the empty set as an element whenever $A$ is a cycle.

\subsection{The introduce bag}
\label{ssec:introduce}

The introduce bag is one of two steps of the algorithm where we can possibly 
conclude that our complex is not tight (the other one being the join bag).

Let $B_I$ be an introduce bag, that is, there is exactly one vertex of 
of the dual graph $\Gamma (M)$ added to its child bag. Hence, there are at
most $(d+1)$ vertices added to the previous induced subcomplex and less than
${ (d+1) (k+1) \choose j+1}$ added $j$-dimensional faces appearing in
$M[V(B_I)]$.

First, let $c$ be a chain in $M_I^-$ with $c \cap M[V(B_i)] = A$ and
$\partial c = b$ as described above. For any $j$-dimensional triple 
$(A,b,\mathcal{C})$ from the list of the previous bag, we enumerate 
$(j+1)$-chains $D$ from newly added $(j+1)$-faces for which either  
$\partial D = C \cup A'$ where $C \in \mathcal{C}$, 
$C \cap A' = b$, and $A'$ is disjoint from $A$, or $\partial D = c \cup A'
= A \cup A'$, with $A'$ as before.
See Figure~\ref{fig:introduce} for a schematic representation of the former
situation.

\begin{figure}
	\begin{center}
	\includegraphics[width=.8\textwidth]{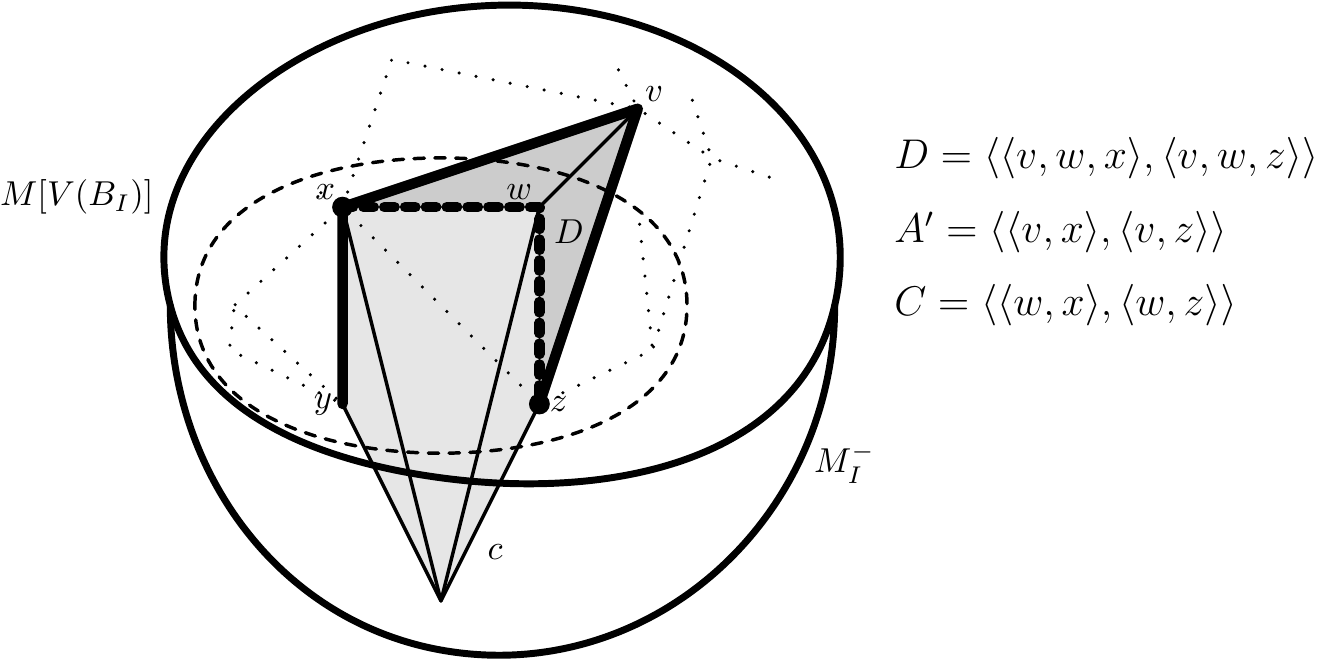}
	\end{center}
	\caption{Disproving $j$-tightness at an introduce bag.
		Note that $c \cup A'$ forms a $1$-cycle.
		However, $c \cup A'$ is also the boundary 
		of the union of the
		two shaded areas which contain the internal vertex $w$.
		\label{fig:introduce}}
\end{figure}

\begin{lemma}
	If at least one such $(j+1)$-chain $D$ exists, and if all such $D$
	either contain interior vertices or have $\partial D = C \cup A'$ where $C$ 
	contains a vertex disjoint from $c$, then $M$ is not tight.
\end{lemma}

\begin{proof}
	Look at the induced subcomplex $M_I^-[V(c \cup A')]$. Clearly, 
	$c \cup A'$ is a $j$-cycle in $M_I^-[V(c \cup A')]$. Let $D'$ be a 
	$(j+1)$-chain in $M_I^-[V(c \cup A')]$ bounding $c \cup A'$. By
	construction, it contains one of the $(j+1)$-chains $D$ from above.
	But $D$ is a chain in $M_I^-[V(c \cup A')]$ so $D$ cannot 
	have interior vertices. Hence, $D'$ must contain an interior $j$-face 
	$C \in \mathcal{C}$. But by assumption,
	this $C$ must then have a vertex disjoint from $c$, 
	cannot be contained in $M_I^-[V(c \cup A')]$, and hence
	such a $D'$ cannot exist and $c \cup A'$ is not a boundary in 
	$M_I^-[V(c \cup A')]$.

	On the other hand, $c \cup A'$ clearly is a boundary in $M_I^-$ and thus
	is a boundary in $M$, and therefore $M$ is not tight.
\end{proof}

Now, if no such $(j+1)$-chains $D$ occur in the newly added faces, 
we extend $A$, update $b$ and all members of $\mathcal{C}$ such that the new 
triples meet all assumptions as explained above in detail. Note that this
can be done by examining only the induced subcomplex $M[V(B_I)]$ and
the list from the child bag.

\subsection{The forget bag}

In the following, we will refer to the
faces removed from the current induced subcomplex as {\em forgotten faces}.
For each triple $(A,b,\mathcal{C})$, we:
\begin{itemize}
	\item Discard all triples where $b$ contains a forgotten $(j-1)$-face. 
	\item Remove all forgotten $j$-faces from $A$. This will possibly 
		leave $A$ empty, which is valid as long as no face of $b$
		is forgotten. 
	\item Remove all entries of $\mathcal{C}$ which contain 
		forgotten $j$-faces.
	\item Merge any list entries that have become identical.
\end{itemize}


\subsection{The join bag}
\label{ssec:join}

In a join bag $B_J$, two lists on the same induced subcomplex $M[V(B_J)]$
are merged together. This can result in new list entries disproving tightness
of $M$.

For each pair of list entries $(A,b,\mathcal{C})$ and $(A',b',\mathcal{C}')$ 
we make list entries $(A'',b'',\mathcal{C}'')$ where $A'' = A + A'$ 
(as an addition of chains over $\mathbb{F}_2$), and $b'' = b + b'$ 
(as an addition of chains over $\mathbb{F}_2$). In addition, we define 
$\mathcal{C}''$ to be the list of all pairs $C'' = C + C'$ 
(as an addition of chains over $\mathbb{F}_2$) with 
$C \in \mathcal{C}$ and $C' \in \mathcal{C}'$ whenever $C''$ and $A''$ 
do not share a $j$-face. 

Clearly, any $C''$ will close off $A''$ along $b''$ and for each 
$C'' \in \mathcal{C''}$ there will exist a 
$(j+1)$-chain in $M_J^-$ bounding 
$c'' \cup C''$ (here $c''$ denotes a chain in $M_J^-$ corresponding to
the sum of two chains $c$ and $c'$ grouped together in the patterns
$A$ and $A'$ in $M[V(B_J)]$).

Now, for each fixed pair of triples 
$(A,b,\mathcal{C})$ and $(A',b',\mathcal{C}')$ with $b + b' = \emptyset$ 
we define $ \mathcal{C}_0'' = \mathcal{C} \cap \mathcal{C}'$.

\begin{lemma}
	If $\mathcal{C}_0'' \neq \emptyset$ and every 
	element in $\mathcal{C}_0''$ contains a vertex disjoint to $A''$
	then $M$ is not tight.
\end{lemma}

\begin{proof}
	$\mathcal{C}_0''$ exactly denotes the combinations of elements
	which result in the empty set in $\mathcal{C}''$ each. Look at the 
	induced subcomplex $M_J^{-} [V(c + c')]$ where $c$ and $c'$ are
	representatives of the $j$-chains in $M_J^-$ with patterns $A$ and
	$A'$ respectively. By construction, $b + b' = \emptyset$ and 
	$c + c'$ is a $j$-cycle in $M_J^{-} [V(c + c')]$. Now, assume that
	there is a $(j+1)$-chain $D$ in $M_J^{-} [V(c + c')]$ bounding $c + c'$.
	By design of the algorithm, $D$ must contain a $j$-chain $C$  
	occurring in $\mathcal{C}_0''$. However, by assumption all such 
	$(j+1)$-chains contain a vertex not in $V(c+c')$ and thus
	$c+c'$ is not a boundary in $M_J^{-} [V(c + c')]$. 
	The proof is completed by the fact that $c+c'$ is a boundary in 
	$M_J^{-}$ and thus $M$ cannot be tight.
	See Figure~\ref{fig:join} for an illustration of such a configuration.

	\begin{figure}
		\begin{center}
		\includegraphics[width=.8\textwidth]{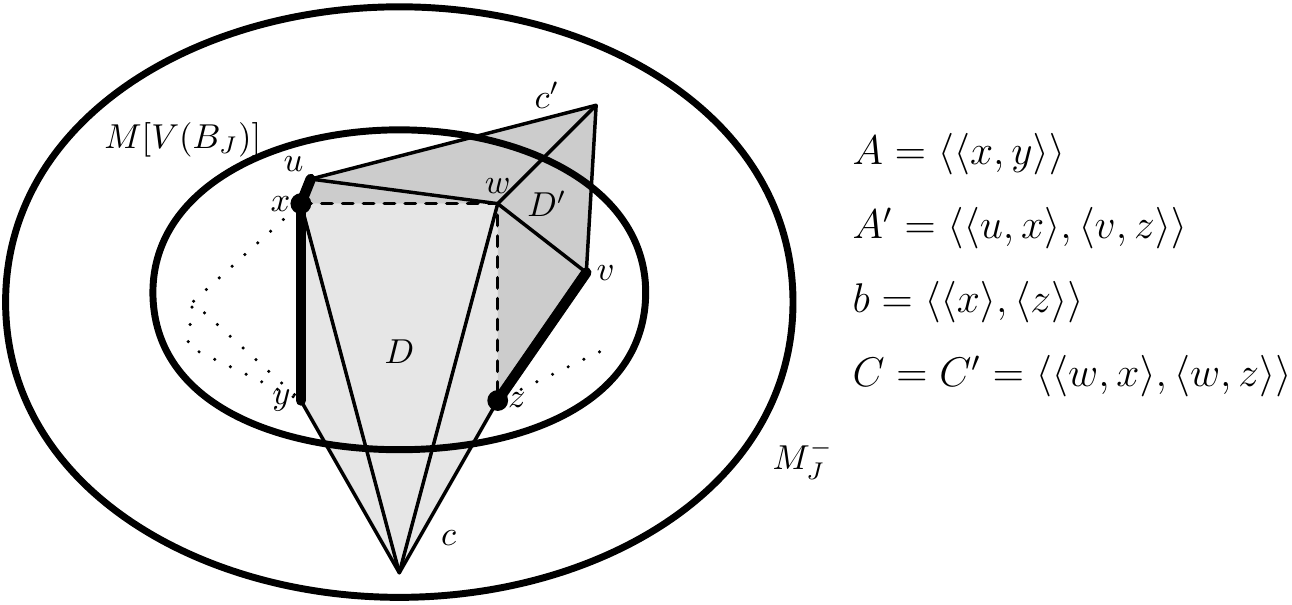}
		\end{center}
		\caption{Disproving $j$-tightness at a join bag.
			Again, the two shaded areas bound a $1$-cycle and
			contain the internal vertex $w$.
			\label{fig:join}}
	\end{figure}
\end{proof}

\subsection{The root bag}

After processing the root bag we will be left with a single list entry
of type $(\emptyset,\emptyset,\emptyset)$. If the algorithm reaches this
step without finding an obstruction to tightness we conclude that
$M$ is tight.

\subsection{Correctness of the algorithm}

Let $(c,D)$ be an obstruction to tightness; that is, $c$ is a $j$-cycle in
$M$ which is not a boundary in $M[V(c)]$, and $D$ is a $(j+1)$-cycle
in $M$ with $\partial D = c$.

We define $W = V(D) \setminus V(c)$ to be the {\em set of internal vertices
of $D$}. Since, by assumption, $D \not \subset M[V(c)]$, $W$ is not empty.

Let $B_i$ be the first bag such that $D' = D \cap M_i^-$ contains an internal 
vertex. Denote this vertex by $v \in W$. It follows that 
$B_i$ is either an introduce bag or a join bag. We will show that
$(\partial D',D')$, which is an obstruction to tightness, will be
identified by the algorithm in both cases.

\medskip
{\bf Case 1}: Let $B_i$ be an introduce bag and let $B_j$ be the unique
bag preceding $B_i$. By construction, the symbol 
$$(A,b,\mathcal{C}) = (c[V(B_j)], \partial (c[V(M_j^-)]), \{ \ldots , \partial D[V(M_j^-)] \setminus c[V(M_j^-)], \ldots \})$$
must occur in the list associated to the bag $B_j$ (note that $A$ is some 
$j$-chain in $M[V(B_j)]$, $b$ is some $(j-1)$-chain contained in the boundary
of $c$ and $\mathcal{C}$ contains all elements which complete the boundary
of some $(j+1)$-chain in $M_j^-$ inside $M[V(B_j)]$).

Now, note that by construction, we can form a $j$-chain $A'$ and
a $(j+1)$-chain $D''$ from faces added by $B_i$ such that
$\partial D'' = A' \cup \partial D[V(M_j^-)] \setminus c[V(M_j^-)]$.
This is exactly the case described in Section~\ref{ssec:introduce} and the
algorithm will terminate stating that $M$ is not tight.

\medskip
{\bf Case 2}: Let $B_i$ be a join bag and let $B_j$ and $B_h$ be the
two bags preceding $B_i$. By the properties of a nice tree decomposition we 
have that $M[V(B_i)] = M[V(B_j)] = M[V(B_h)]$. Since $v$ is internal in
$D'$ but does not occur as an internal vertex in any other preceding bag,
the lists of $B_j$ and $B_h$ must have complementary entries which, combined,
have a cycle as $A$ and vanishing $b$ and an empty element in $\mathcal{C}$. This is 
exactly the kind of situation described in Section~\ref{ssec:join} and the
algorithm will terminate stating that $M$ is not tight.

\subsection{Running time}

We run the algorithm sequentially for each $1 \leq j \leq d-1$.

In each step, $M[V(B_i)]$ has at most $2^{(d+1)(k+1)}$ faces. It follows, that
each list has at most $2^{2^{(d+1)(k+1)}}$ entries of type $(A,b,\mathcal{C})$
and the list $\mathcal{C}$ has less than $2^{2^{(d+1)(k+1)}}$ entries. 
Every list entry, or every pair of list entries can be 
updated and checked in polynomial time in the size of the list entry.
Note that look-ups and merge operations are logarithmic in the list size
since every list is sorted by the keys $(A,b)$ each pointing to the
list $\mathcal{C}$ which is itself sorted.

It follows that the running time of the join bag dominates the overall running
time, with a pairwise merging of lists requiring time
$\left (2^{2^{(d+1)(k+1)}} \right)^2$.
Each bag has to be visited exactly once and there are $O(n)$ bags, where $n$ is
the number of facets in $M$. Hence, the overall running
time of the algorithm can be bounded by 
$O\left (n \cdot \left (2^{2^{(d+1)(k+1)}} \right)^2 \operatorname{poly} \left (2^{(d+1)(k+1)} \right ) \right)$.

Theorem~\ref{thm:fptd} now follows.

	{\footnotesize
	 \bibliographystyle{plain}
	 \bibliography{biblio}
	}

\end{document}